\documentclass[a4paper,USenglish,cleveref,thm-restate]{lipics-v2021}
\hideLIPIcs
\nolinenumbers
\usepackage[textsize=footnotesize]{todonotes}
\usepackage{amsmath,amssymb,amsthm}
\usepackage{booktabs}
\usepackage{microtype}
\usepackage{tikz}
\usepackage{csquotes}
\usepackage{tcolorbox}

\usepackage{mathtools}
\DeclarePairedDelimiter{\ceil}{\lceil}{\rceil}
\DeclarePairedDelimiter{\floor}{\lfloor}{\rfloor}

\usetikzlibrary{arrows.meta,calc,positioning}

\usepackage{bm}

\usepackage[bibliography=common]{apxproof}

\newcommand{\problemdec}[3]{
	\begin{tcolorbox}[width = \textwidth,colback=white,arc=0pt,outer arc=0pt,boxrule=0.7pt,left =0.5em,right=0em]#1		\\[2pt]
		\begin{tabular}{ @{}l p{0.84\textwidth} c }
			\textbf{Input:} & #2 \\[.5pt]
			\textbf{Task:} & #3
		\end{tabular}
	\vspace{-0.25em}
	\end{tcolorbox}
}

\DeclareMathOperator{\tw}{tw}
\DeclareMathOperator{\pw}{pw}
\DeclareMathOperator{\mg}{mg}
\DeclareMathOperator{\rad}{rad}
\DeclareMathOperator{\dist}{dist}
\DeclareMathOperator{\colors}{Colors}
\newtheorem{openproblem}[theorem]{Open Problem}

\title{Polynomial pathwidth for graphs vertex-coverable by few shortest paths and algorithms for quasi-shortest path covers}
\titlerunning{Polynomial pathwidth and algorithms for graphs coverable by few shortest paths}

\author{Li-Hsuan Chen}
{Department of Computer Science and Information Engineering, National Chung
 Cheng University, Chiayi, Taiwan}
{lhchen@cs.ccu.edu.tw}{https://orcid.org/0009-0007-6528-7064}{}
\author{Florent Foucaud}{Université Clermont Auvergne, CNRS, Mines Saint-Étienne, Clermont Auvergne INP, LIMOS, 63000 Clermont-Ferrand, France\and \url{https://perso.limos.fr/ffoucaud}}{florent.foucaud@uca.fr}{https://orcid.org/0000-0001-8198-693X}{ANR project GRALMECO (ANR-21-CE48-0004), French government IDEX-ISITE initiative 16-IDEX-0001 (CAP 20-25), International Research Center ``Innovation Transportation and Production Systems'' of the I-SITE CAP 20-25.}
\author{Ling-Ju Hung}{Department of Computer Science and Information Engineering, National Taipei University of Business, Taipei, Taiwan}{ljhung@ntub.edu.tw}{https://orcid.org/0000-0001-5659-5507}{}
\author{Ralf Klasing}{Universit\'e de Bordeaux, CNRS, LaBRI, Bordeaux, France\and\url{https://www.labri.fr/perso/klasing}}{ralf.klasing@labri.fr}{https://orcid.org/0009-0004-9421-8525}{ANR project TEMPOGRAL (ANR-22-CE48-0001).}
\author{Lucas Lorieau}{Université Clermont-Auvergne, CNRS, Mines de Saint-Étienne, Clermont-Auvergne-INP, LIMOS, 63000 Clermont-Ferrand, France\and \url{https://perso.limos.fr/~lulorieau/}}{lucas.lorieau@limos.fr}{https://orcid.org/0009-0006-5084-5118}{This author has received financial support from the CNRS through the MITI interdisciplinary programs.}
\author{Peter Rossmanith}{RWTH Aachen University, Aachen, Germany}{rossmani@cs.rwth-aachen.de}{https://orcid.org/0000-0003-0177-8028}{}

\authorrunning{L.-H. Chen, F. Foucaud, L.-J. Hung, R. Klasing, L. Lorieau,
 and P. Rossmanith}
\Copyright{Li-Hsuan Chen, Florent Foucaud, Ling-Ju Hung, Ralf Klasing,
 Lucas Lorieau, and Peter Rossmanith}


\funding{This research was partially supported by the NSTC--CNRS International Emerging Action (IEA) 2024--2025 ``Quasi-Polynomial Time Approximation Algorithms for Some Network Design Optimization Problems'' under grants 113--2927--I--141--501 and 114--2927--I--141--501.}
\acknowledgements{We thank Philip Whittington for bringing the metric-growth
 theorem of Campbell et al.~\cite{CampbellDGGHHMW2023} to our attention.}

\begin{document}
\maketitle

\begin{abstract}
We study how covering a graph by a small number of shortest or approximately shortest paths/walks constrains its treewidth and pathwidth. A path or walk is $(\alpha,\beta)$-quasi-shortest for $\alpha\geq1$ and $\beta\geq 0$ if each contiguous segment has length at most $\alpha$ times the distance between its endpoints plus $\beta$. This relaxation forms a reasonable model of approximate shortest paths permitting local detours, closer to applications.

For graphs vertex-covered by $k$ shortest paths, an exponential upper bound on pathwidth was known [Dumas et al., ISAAC 2022/SIDMA 2024]. We establish the first known polynomial pathwidth bound of $O(k^3)$. For edge-coverage, we improve the previously known $O(k^4)$ upper bound [Baste et al., STACS 2026] to $O(k^2\log k)$. As a main ingredient of the proofs, we first show that if the vertices of a graph can be covered by $k$ $(\alpha,\beta)$-quasi-shortest paths or walks, then its treewidth is $O_{\alpha,\beta}(k^2)$; this is a consequence of the fact that such graphs have linear metric growth and the technique developed in [Campbell et al., E-JC 2023]. To obtain the above pathwidth upper bounds, we combine this result with techniques from the literature in a nontrivial way. The metric growth toolkit also implies improved upper bounds for $H$-minor-free graphs.

On the other hand, we show that there exist connected bipartite cubic graph whose edges can be covered by at most $k$ shortest paths and whose treewidth and pathwidth are in $\Omega(k\log k)$, providing the first superlinear lower bound for the problem. For graphs covered by $k=2$ $(\alpha,0)$-quasi-shortest paths with $\alpha>1$, we show that, unlike treewidth and unlike the case $\alpha=1$ and $\beta=0$ of exact shortest paths, the pathwidth may be unbounded.

Finally, for every fixed rationals $\alpha\geq1$ and $\beta\geq 0$, we give XP algorithms with a running time of $n^{f(k,\alpha,\beta)}$ for some function~$f$ that decide whether the vertices of a graph can be covered by at most $k$ $(\alpha,\beta)$-quasi-shortest paths or walks and return such a family whenever one exists. The algorithms combine our treewidth upper bounds with exact evaluation of Blocked-CMSO queries using the algorithmic meta-theorem from [Dreier et al., LICS 2025]. These algorithms extend a similar result for exact shortest paths from [Dumas et al., ISAAC 2022/SIDMA 2024].
\end{abstract}

\section{Introduction}\label{sec:introduction}

Shortest paths in graphs (also called \emph{geodesics} or \emph{isometric paths}) are a fundamental concept in graph theory and computer science. We are interested in the problem of covering the vertices or the edges of a given graph by a small number of shortest paths, a natural optimization problem that recently attracted a lot of attention~\cite{BasteDMGOP2026,DBLP:journals/dm/ChakrabortyCFV26,DBLP:conf/isaac/ChakrabortyD0FG22,ChakrabortyDFMT2026,DBLP:journals/tcs/ChakrabortyMOPR26,DumasFPT2024}. Such covers provide a simplified model for network design problems of planning small collections of transportation routes that collectively visit the network while limiting detours, such as bus lines~\cite{BertsimasNgYan2021} or patrolling rounds~\cite{ChenChengShaweTaylor2017}. Shortest-path covers also arise in active learning on graphs, where geodesic convexity allows labels along each covering path to be recovered through binary search~\cite{ThiessenG2021}. It was recently shown that a graph whose vertices can be covered by few shortest paths has bounded (exponential) pathwidth~\cite{DumasFPT2024}. Subsequently, a polynomial bound when also the edges are covered was proved~\cite{BasteDMGOP2026}. Whether polynomial bounds also hold for vertex-coverage remained open.

We generalize this line of work by providing the first polynomial upper bound on the pathwidth of a graph when the vertices are covered (and no restriction on also covering the edges). This settles the main open problem in the area~\cite{BasteDMGOP2026,DumasFPT2024}. We also give several improved upper bounds for the edge-covering case, for graphs excluding a fixed minor, or when the treewidth is considered, as well as a nontrivial lower bound. Motivated by real-world scenarios, we then relax the requirement of considering shortest paths only, and consider settings where the paths are ``almost shortest'' paths or walks (allowing a fixed multiplicative and additive distortion for every segment of the path/walk). We show that treewidth is bounded in this setting as well, however pathwidth is not. We use this to generalize the existing algorithmic results for covering graphs by shortest paths~\cite{DumasFPT2024}, providing XP algorithms for covering a graph in this more realistic scenario. (We use the standard terminology that an FPT running time has the form $f(k)n^{O(1)}$, while an XP running time has the form $n^{f(k)}$, up to parameter-dependent factors~\cite{CyganFKLMPPS2015}; $n$ is the input size, $k$ the parameter, and $f$ a computable function.)

\medskip
\noindent\textbf{Quasi-shortest paths for traffic connections and patrol rounds.} Requiring shortest paths can be unnecessarily rigid in applications: suppose for example that the paths represent bus or patrol routes: it is unrealistic that every itinerary is precisely a shortest path, and small detours can be allowed in order to reduce the total number of routes~\cite{BertsimasNgYan2021}. Similarly, it may be useful to allow crossing the same location several times, and to consider \emph{walks} instead of paths~\cite{ChenChengShaweTaylor2017}: consider
a long path with an additional leaf $x$ adjacent to an internal vertex $v$.
The resulting tree has three leaves, so no single simple path covers all its
vertices. However, one patrol officer can
instead follow the long path and, on reaching $v$, make the short excursion
to $x$, before continuing along the path. We therefore
consider a relaxation of the problem, similar to the definition of quasi-geodesics from coarse geometry and geometric group theory (and related to the quasi-isometries from these fields, also transposed to graphs in coarse graph theory), as follows.

For fixed reals $\alpha\geq 1$ and
$\beta\geq0$, a path or walk in a graph $G$ with vertex occurrences $v_1\ldots v_t$ is \emph{$(\alpha,\beta)$-quasi-shortest}
if for any $i,j$ with $i<j$, 
$$j-i\leq \alpha\cdot\dist_G(v_i,v_j)+\beta\;.$$

If $\alpha=1$ and $\beta=0$, the above definition simply describes shortest paths. Unlike a guarantee that applies only to the route's endpoints, the above condition
imposes a distortion for every segment of the path/walk. This definition is inspired by classic works in coarse geometry and geometric group theory on quasi-geodesics and quasi-isometries~\cite{BridsonHaefliger1999,Gromov1993}. 
The notion of a quasi-isometry was recently extended to coarse graph theory~\cite{Georgakopoulos2025,nguyen2025coarsetreewidth}, following earlier works on the similar notion of distance-approximating embeddings~\cite{ChepoiDraganEstellonHabibVaxes2008}. The definition of quasi-shortest paths we use is also similar to the notion of \emph{$(\alpha,\beta)$-spanners}; such a spanner is a spanning subgraph $H$ of a host graph $G$ such that, for every pair $x,y$ of vertices, $\dist_H(x,y)\leq\alpha\cdot\dist_G(x,y)+\beta$. The main difference is that in a spanner, the distortion condition may not apply to all $x$--$y$ paths, but only to the shortest ones. Less constrained distortion conditions, allowing a multiplicative distortion factor for the whole path, have been used in the context of approximate all-pairs-shortest paths~\cite{DBLP:conf/soda/SahaY24}. \emph{Local geodesics} (where every small subpath must be a shortest path, but without a global condition) are also considered in discrete geometry~\cite{ChalopinChepoiNaves2015}. These two conditions (global distortion together with local optimality) have been combined in works considering \emph{alternative routes}~\cite{DBLP:journals/jea/AbrahamDGW13}.

\medskip
\noindent\textbf{Further related work.} Covering a graph by a minimum number of (unrestricted) paths is a fundamental and extensively studied problem, called \textsc{Path Cover}~\cite{dagPC,PC-cocomp,survey}; think of \textsc{Hamiltonian Path} as a prominent special case. The shortest path requirement that we consider is substantially stronger than an ordinary path cover: a Hamiltonian graph is covered by one path but may require arbitrarily many shortest paths to be fully covered (consider for example a square grid or a complete graph).

The problem of covering a graph by a minimum number of shortest paths arose in the context of the \emph{Cops and Robber} pursuit-evasion game~\cite{AF84}. Covering by shortest paths still plays an important role in more modern techniques in the area, see for example the concept of cop-decompositions, which are special tree decompositions whose definition involves shortest path covers~\cite{cop-decs}. The problem was also studied for its own sake from a graph-theoretic perspective, see for example~\cite{FF01,FNHC01,DBLP:journals/dm/PanC06}, and refer to the survey~\cite{survey} on various types of path covering problems.

Dumas, Foucaud, Perez, and Todinca~\cite{DumasFPT2024} proved that a graph
vertex-covered by $k$ shortest paths has BFS-width (the maximum size of a BFS distance-layer~\cite{DBLP:conf/esa/EppsteinGL25}), and thus pathwidth, $O(k\cdot3^k)$, while the more constrained setting of also covering the edges
gives $O(3^k)$~\cite{DumasFPT2024}. 

Baste, De Meyer, Giocanti, Objois, and Picavet~\cite{BasteDMGOP2026}
subsequently proved that a graph edge-covered by $k$ shortest paths has
pathwidth $O(k^4)$. More precisely, their proof constructs a polynomial-size
vertex set whose deletion leaves pathwidth $O(k)$. They asked whether the
pathwidth might in fact be linear in $k$.

The (non-)existence of similar bounds when the graph is covered by isometric \emph{trees} was also investigated, see~\cite{BasteDMGOP2026,bastide2025counterexamplesstatementsisometricgraph}.

Hatzel and Pilipczuk~\cite{HatzelP2026} study a relaxation of the question in the realm of \emph{coarse graph theory}:
the $k$ routes remain shortest paths, but it is enough for every vertex
to lie within distance $\rho$ of one of them. They proved that such graphs are $(3,12\rho)$-quasi-isometric to graphs of pathwidth $k^{O(k)}$. 

The problem of covering a graph with a minimum possible number of exact shortest paths (often called \textsc{Isometric Path Cover}~\cite{DBLP:journals/dm/ChakrabortyCFV26,survey} or \textsc{Shortest Path Cover}~\cite{FernauFMMPR2025,ThiessenG2021}) is NP-hard, even for chordal graphs of diameter~2~\cite{DBLP:conf/isaac/ChakrabortyD0FG22} and 5-degenerate bipartite graphs of diameter~4~\cite{FernauFMMPR2025}. On the other hand, it is polynomial-time solvable on block graphs~\cite{IPC-block}, chain graphs~\cite{DBLP:journals/tcs/ChakrabortyMOPR26} and cographs~\cite{DBLP:journals/tcs/ChakrabortyMOPR26}. FPT algorithms for the neighborhood diversity and the dual parameter $n-k$ were designed in~\cite{FernauFMMPR2025}.

In the fixed-terminal version of \textsc{Shortest Path Cover}, the $k$ endpoint pairs are
part of the input, and the authors in~\cite{DumasFPT2024} give an FPT algorithm in $k$ for deciding whether the input graph is coverable by $k$ shortest paths. When the algorithm
must also choose the endpoints, enumerating the $2k$ endpoints and invoking
the fixed-terminal algorithm gives time $f(k)n^{2k+O(1)}$, hence XP~\cite{DumasFPT2024}.

On chordal graphs (and more generally graphs of bounded treelength~\cite{DBLP:conf/isaac/ChakrabortyD0FG22} or hyperbolicity~\cite{DBLP:journals/dm/ChakrabortyCFV26}), on outerstring graphs~\cite{DBLP:journals/dm/ChakrabortyCFV26} and on (theta,prism,pyramid)-free graphs~\cite{DBLP:journals/dm/ChakrabortyCFV26}, \textsc{Shortest Path Cover} was shown to admit a constant-factor approximation algorithm based on Dilworth's theorem on poset chain coverings. For a connected graph, write $\operatorname{diam}(G)$ for the maximum distance
between two vertices. Thiessen and G\"artner~\cite[Theorem~11]{ThiessenG2021}
proved that \textsc{Shortest Path Cover} can be approximated within a factor of
$1+\ln(\operatorname{diam}(G)+1)$ in $O(|V(G)|^4)$ time.  

The version where the paths must be vertex-disjoint, which is used in the recent definition of \emph{geodesic treewidth}~\cite{geodesictreewidth} (which minimizes the treewidth of the graph obtained by contracting each path of a shortest path partition), has also been studied~\cite{ChakrabortyDFMT2026,FernauFMMPR2025,Manuel2021}. Notably, this variant is W[1]-hard when parameterized by pathwidth, but open for solution size~\cite{ChakrabortyDFMT2026}. The edge-disjoint variant is W[1]-hard on \emph{oriented} graphs (in fact, DAGs) when parameterized by solution size~\cite{FernauFMMPR2025}.

\medskip
\noindent\textbf{Our results.} Our main results concern improved upper and lower bounds on treewidth and pathwidth for graphs coverable by $k$ (quasi-)shortest paths, which answer open problems from~\cite{BasteDMGOP2026,DumasFPT2024}. We first prove a quadratic $O(k^2)$ bound on the treewidth, which we then use to obtain our upper bounds on pathwidth.

To this end, we observe that the concept of \emph{metric growth} (see Section~\ref{sec:preliminaries} for the definition) can be applied to such graphs: roughly speaking, the maximum intersection of any quasi-shortest path with any $r$-ball grows only linearly with the radius $r$. Hence, graphs whose vertices are coverable by $k$ $(\alpha,\beta)$-quasi-shortest paths have metric growth $O_{\alpha,\beta}(k)$, which directly implies treewidth $O_{\alpha,\beta}(k^2)$ using a result from Campbell et al.~\cite{CampbellDGGHHMW2023}. This bound is reduced to $O_{\alpha,\beta,H}(k)$ for graphs excluding a graph $H$ as a minor.

Then, we focus on pathwidth and on exact shortest paths. We use the above treewidth upper bound in conjunction with the technique of Dumas et al.~\cite{DumasFPT2024} to obtain an $O(k^3)$ upper bound on pathwidth when the vertices (and not necessarily the edges) of the graph are covered by $k$ exact shortest paths. This shows that also pathwidth is polynomially bounded in this case, and solves the main open problem in the area. 

We next improve the upper bound in the case where the $k$ shortest paths cover the edges. To this end, we prove a general conversion lemma: if $\tw(G)\leq w$ and deleting $r$ vertices leaves a graph with pathwidth at most $p$, then $\pw(G)=O((w+p)\log r)$. We apply this lemma to the pathwidth-modulator of Baste et al.~\cite{BasteDMGOP2026} and our treewidth upper bound, to yield the quadratic-logarithmic bound $O(k^2\log k)$ on pathwidth. This improves the previous $O(k^4)$ bound from Baste et al.~\cite{BasteDMGOP2026}. For graphs excluding a graph $H$ as a minor, we obtain $O_H(k\log k)$.

Regarding lower bounds, we prove that both the vertex-cover and the edge-cover pathwidth and treewidth upper bounds cannot be linear. Indeed, we show that certain cubic expanders of order $n$ with large girth have treewidth $\Omega(n)$ and can be edge-covered by $O(n/\log n)$ shortest paths. This implies an $\Omega(k\log k)$ treewidth (and pathwidth) lower bound, which is the first known superlinear lower bound. This answers negatively a question from Baste et al.~\cite{BasteDMGOP2026}.

When $\alpha>1$, we also provide a construction of graphs covered by two $(\alpha,0)$-quasi-shortest paths but with arbitrarily large pathwidth, showing that the quasi-shortest setting behaves differently than the exact shortest setting.

A summary of the bounds we obtain can be found in \Cref{tab:results}. 

\begin{table}[ht]
\caption{Summary of our bounds for general graphs coverable by $k$ (quasi-)shortest paths.} 
\label{tab:results}
\begin{tabularx}{\textwidth}{@{}
 >{\raggedright\arraybackslash}X
 >{\raggedright\arraybackslash}p{.21\textwidth}
 >{\raggedright\arraybackslash}X
 >{\centering\arraybackslash}X
 >{\centering\arraybackslash}X
 @{}}
\toprule
coverage & covering objects & lower bound\footnotemark[1] & \multicolumn{2}{c}{upper bound\footnotemark[2]} \\
\cmidrule(lr){4-5}
 & & & treewidth & pathwidth \\
\midrule
vertices/edges & $(\alpha,\beta)$-quasi-shortest paths/walks
 & $\Omega(k\log k)$  & $O_{\alpha,\beta}(k^2)$ & unbounded in $k$ when $\alpha>1$\\[6mm]
vertices & exact shortest paths
 & $\Omega(k\log k)$
 & $O(k^2)$
 & $O(k^3)$ \\[2mm]
edges & exact shortest paths
 & $\Omega(k\log k)$
 & $O(k^2)$
 & $O(k^2\log k)$ \\
\bottomrule
\end{tabularx}

\scriptsize\smallskip
${}^1$ The lower bounds hold for both treewidth and pathwidth. See \Cref{thm:edge-lower}.\\
${}^2$ See \Cref{thm:treewidth,cor:exact-treewidth,thm:walk-treewidth,cor:minor-closed-treewidth,thm:edge-upper,thm:twopaths}.
\end{table}

Then, we apply the treewidth upper bound for quasi-shortest paths to obtain XP algorithms for \textsc{$(\alpha,\beta)$-Quasi-Shortest Path Cover} and \textsc{$(\alpha,\beta)$-Quasi-Shortest Walk Cover} (formally defined in \Cref{sec:preliminaries}) parameterized by the number $k$ of paths. The main algorithmic ingredient is the logic introduced by Dreier, Ganian, and Hamm~\cite{DreierGH2025}. Their \emph{Blocked CMSO} permits restricted cardinality comparisons under universal quantification and admits both approximate and exact evaluation theorems on bounded-treewidth graphs. We use their theorem as a black-box. The exponent of the running time depends on the formula and the width. The former depends on $k$, $\alpha$ and $\beta$, and the latter is in $O_{\alpha,\beta}(k^2)$; hence, we obtain a running time that is XP in $k$ for fixed rationals $\alpha$ and $\beta$. 
These algorithms generalize the XP algorithm for exact shortest paths obtained by Dumas et al.~\cite{DumasFPT2024}, to the setting of quasi-shortest paths and walks.

\medskip
\noindent\textbf{Outline.} After the preliminaries, \Cref{sec:treewidth} contains the treewidth upper bounds, \Cref{sec:vertex-upper-pathwidth} the pathwidth upper bound for the vertex-covering case, \Cref{sec:edge-upper} proves the pathwidth upper bound for the edge-covering case, \Cref{sec:edge-lower} contains the pathwidth lower bound for exact shortest paths, \Cref{sec:quasi-shortest-unbounded} contains the unbounded pathwidth construction for quasi-shortest paths, and \Cref{sec:blocked} 
gives the XP algorithms. We conclude in \Cref{sec:discussion}. 

\section{Preliminaries}\label{sec:preliminaries}

All graphs are finite, simple, undirected, and unweighted.  We allow paths
of length zero, consisting of one vertex.  Every route budget $k$ is a
positive integer. 
For vertices $x,y$ of a graph $G$, let
$\dist_G(x,y)$ denote the minimum number of edges in an $x$--$y$ path. 
A path is \emph{shortest} if its length equals the distance between its
endpoints; every subpath of a shortest path is shortest.

A \emph{walk} is a sequence  $W=w_0e_1w_1\cdots e_tw_t$, where $e_j=w_{j-1}w_j$.  Its \emph{length} is $t$, counting edge occurrences, so
repeated uses of an edge are counted separately.  A \emph{subwalk}
$W[a,b]$ is the contiguous sequence from the occurrence $w_a$ to $w_b$.
A \emph{path} is a walk with pairwise distinct vertices.

\begin{definition}\label{def:quasi-shortest}
Let $\alpha\geq1$ and $\beta\geq0$. A walk $W$ in $G$ is
\emph{$(\alpha,\beta)$-quasi-shortest} if every subwalk $W[a,b]$ satisfies $b-a\leq \alpha\cdot\dist_G(w_a,w_b)+\beta.$
An \emph{$(\alpha,\beta)$-quasi-shortest path} is an $(\alpha,\beta)$-quasi-shortest walk that is a simple path.
\end{definition}

A path is shortest precisely when it is $(1,0)$-quasi-shortest.  A
$(1,0)$-quasi-shortest walk cannot repeat a vertex, because the subwalk
between two consecutive occurrences would have positive length and endpoint
distance zero. Thus, it is also a shortest path.

A family $\mathcal P$ of paths is called a \emph{vertex path cover} if
$V(G)=\bigcup_{P\in\mathcal P}V(P)$, and an \emph{edge path cover} if
$E(G)=\bigcup_{P\in\mathcal P}E(P)$.  Different paths may share vertices
and edges.  A \emph{vertex walk cover} is defined analogously using the
distinct vertices visited by each walk. For fixed $\alpha\geq 1$ and $\beta\geq 0$ ($\alpha,\beta$ may be real, or rational for practical purposes), we define the two following problems.

\problemdec{\textsc{$(\alpha,\beta)$-Quasi-Shortest Path Cover}}{A graph $G$ and an 
integer \(k\).}{Determine whether at most
$k$ $(\alpha,\beta)$-quasi-shortest paths cover $V(G)$. For positive instances, return the path cover.}

\problemdec{\textsc{$(\alpha,\beta)$-Quasi-Shortest Walk Cover}}{A graph $G$ and an 
integer \(k\).}{Determine whether at most
$k$ $(\alpha,\beta)$-quasi-shortest walks cover $V(G)$. For positive instances, return the walk cover.}

For a connected graph $H$, let $\rad(H)=\min_{z\in V(H)}\max_{v\in V(H)}\dist_H(z,v)$ be the \emph{radius} of $H$. The \emph{metric growth} of a graph $G$ is the following function $\mg_G:\mathbb{N}\to\mathbb{N}$: $\mg_G(r)=\sup\{\,|V(H)|\mid H\subseteq G,\ H\text{ connected}, \rad(H)\leq r\,\}$. For $v\in V(G)$ and $r\geq0$, we denote the \emph{ball} of radius $r$ centered at $v$ by $B_G(v,r)=\{\,u\mid\dist_G(u,v)\leq r\,\}$.

A \emph{tree decomposition} of $G$ is a tree $T$ with bags
$(B_t)_{t\in V(T)}$ such that every graph vertex is in a bag, the endpoints
of every graph edge occur together in a bag, and the bags containing any
fixed graph vertex induce a connected subtree of $T$.  Its width is
$\max_t|B_t|-1$, and the \emph{treewidth} $\tw(G)$ is the minimum width.
A \emph{path decomposition} is a tree decomposition whose decomposition tree
is a path; its minimum width is the \emph{pathwidth} $\pw(G)$.  Both
parameters are subgraph-monotone, and $\tw(G)\leq\pw(G)$. These are the standard width parameters; see, for example,
Cygan et al.~\cite{CyganFKLMPPS2015}. For the sake of clarity, we will sometimes consider that a path decomposition is an ordered sequence of bags. In this case, the \emph{concatenation} of two path decompositions is simply the concatenation of the two sequences of bags. Note that in general, the concatenation of two path decomposition does not yield a path decomposition.

A graph $H$ is a \emph{minor} of $G$ if it can be obtained from $G$ by
vertex- and edge-deletions and edge-contractions.

\section{Treewidth upper bounds}\label{sec:treewidth}

We show next how \cref{thm:growth} yields a polynomial treewidth upper bound for our problem, when few paths cover the vertices of a graph $G$.  (Notice that when the paths cover $E(G)$, except isolated vertices, they also cover $V(G)$ and all results from this section apply as well.) We will use the following theorems relating metric growth and treewidth.

\begin{theorem}[{Campbell, Distel, Gollin, Harvey, Hendrey,
Hickingbotham, Mohar, Wood~\cite[Theorem~1]{CampbellDGGHHMW2023}}]\label{thm:growth}
For any graph $G$, if there exists a $c\geq 1$ such that the metric growth $\mg_G$ of $G$ satisfies $\mg_G(r)\leq cr$ for every integer $r\geq1$, then $\tw(G)\leq 49c^2+30c$.
\end{theorem}

The same authors also used grid minors to prove the following.

\begin{theorem}[{Campbell, Distel, Gollin, Harvey, Hendrey,
Hickingbotham, Mohar, Wood~\cite[Corollaries~10 and 11]{CampbellDGGHHMW2023}}]\label{thm:growth-minor}
For any fixed graph $H$ and any graph $G$ excluding $H$ as a minor, if there exists an integer $c>1$ such that the metric growth $\mg_G$ of $G$ satisfies $\mg_G(r)\leq cr$ for every integer $r\geq1$, then $\tw(G)= O_H(c)$. In particular, if $G$ is planar, then $\tw(G)\leq 12c+1$.
\end{theorem}

The key element to apply these results to our setting is the following lemma.

\begin{lemma}\label{lem:ball}
Let $P$ be an $(\alpha,\beta)$-quasi-shortest path in $G$.  For every
$z\in V(G)$ and $r\geq0$, $|V(P)\cap B_G(z,r)| \leq 2\alpha r+\beta+1$.
\end{lemma}

\begin{proof}
If the intersection between $V(P)$ and $B_G(z,r)$ is nonempty, let $x$ and $y$ be its first and last
vertices in the order along $P$, and let $s=|E(P[x,y])|$. 
Quasi-shortness and the triangle inequality give
$s\leq \alpha\cdot\dist_G(x,y)+\beta
 \leq\alpha(\dist_G(x,z)+\dist_G(z,y))+\beta
 \leq 2\alpha r+\beta$.
Every vertex of $P$ in $B_G(z,r)$ lies on $P[x,y]$, which has $s+1\leq 2\alpha r+\beta +1$ vertices, and the bound follows.
\end{proof}

We now prove the main result of this section, which is a direct corollary of \cref{lem:ball} and \cref{thm:growth}
by Campbell, Distel, Gollin, Harvey, Hendrey,
Hickingbotham, Mohar, and Wood~\cite[Theorem~1]{CampbellDGGHHMW2023}.

\begin{theorem}\label{thm:treewidth}
Let $\alpha\geq1$ and $\beta\geq0$, and set $
 a=2\alpha+\beta+1$. 
 If $V(G)$ is covered by $k$ $(\alpha,\beta)$-quasi-shortest paths,
then $\tw(G)\leq49a^2k^2+30ak.$
\end{theorem}

\begin{proof}
Let $H\subseteq G$ be connected with radius at most $r\geq1$, and choose a
center $z$ of $H$.  Distances in $G$ can only be shorter than in $H$, so
$V(H)\subseteq B_G(z,r)$.  If $P_1,\ldots,P_k$ cover $V(G)$, then
\cref{lem:ball} and $r\geq1$ give $|V(H)|
 \leq\sum_{i=1}^k|V(P_i)\cap B_G(z,r)|
 \leq akr$. 
Thus, the metric growth of $G$ satisfies $\mg_G(r)\leq akr$, and
\cref{thm:growth} with $c=ak$ proves the claim.
\end{proof}

\begin{corollary}\label{cor:exact-treewidth}
If $V(G)$ is covered by $k$ shortest paths, then $\tw(G)\leq441k^2+90k$.
\end{corollary}

\begin{proof}
A shortest path is $(1,0)$-quasi-shortest, so we can apply
\cref{thm:treewidth} with $a=3$.
\end{proof}

We next show how to apply the same proof idea to walks.

\begin{theorem}\label{thm:walk-treewidth}
Let $\alpha\geq1$ and $\beta\geq0$, and set $
 a=2\alpha+\beta+1$. 
 If $V(G)$ is covered by $k$
$(\alpha,\beta)$-quasi-shortest walks, then $\tw(G)\leq49a^2k^2+30ak.$
\end{theorem}

\begin{proof}
We extend the proof of \cref{lem:ball} and \Cref{thm:treewidth}. Let $z\in V(G)$ and $r\geq 0$ and consider one of the $k$ covering walks $W$. Take the first and last occurrences
of vertices of $W$ in $B_G(z,r)$ and let $s$ be the length of the subwalk between them.
Every distinct vertex of $V(W)\cap B_G(z,r)$ occurs in this interval.  If
$x,y$ are the endpoint vertices, then as in \cref{lem:ball}, $s\leq\alpha\cdot\dist_G(x,y)+\beta\leq2\alpha r+\beta$.
Thus, the bound of \cref{lem:ball} holds for the distinct vertices of every
covering walk. Repeating the proof of \cref{thm:treewidth} gives
$\mg_G(r)\leq akr$ and hence the claimed bound by \cref{thm:growth}.
\end{proof}

We can also use \cref{thm:growth-minor} to obtain a linear upper bound for graphs excluding a fixed minor, in particular planar graphs.

\begin{corollary}\label{cor:minor-closed-treewidth}
Let $H$ be a fixed graph, $\alpha\geq1$, $\beta\geq0$, and $a=2\alpha+\beta+1$. 
If an
$H$-minor-free graph $G$ is vertex-covered by $k$
$(\alpha,\beta)$-quasi-shortest paths or walks, then
$\tw(G)=O_H(ak)$. In particular, if $G$ is planar and is vertex-covered by $k$ exact shortest paths, then $\tw(G)\leq 36k+1$. 
\end{corollary}

\begin{proof}
The proofs of \cref{thm:treewidth,thm:walk-treewidth} give
$\mg_G(r)\leq akr$ and $\mg_G(r)\leq 3kr$ for exact shortest paths. Thus, the bounds follow immediately from \cref{thm:growth-minor}.
\end{proof}

\section{Polynomial upper bound for pathwidth when covering the vertices}\label{sec:vertex-upper-pathwidth}

In this section, we improve the existing upper bound for pathwidth when shortest paths cover the vertex set of the input graph $G$, deriving the first polynomial upper bound. Note that \cref{thm:growth} cannot extend to pathwidth~\cite[Section 5]{CampbellDGGHHMW2023}, so the method used for proving \cref{thm:treewidth} cannot be used to directly obtain a similar bound for pathwidth. We combine
\cref{cor:exact-treewidth} with an extension of the construction of
Dumas et al.~\cite[Section~4]{DumasFPT2024}. The results in this section apply only to shortest paths and not to quasi-shortest ones.

A \emph{layering} of a graph $G$ is a function
$\lambda:V(G)\to\mathbb Z$ such that
$|\lambda(u)-\lambda(v)|\leq1$ for every edge $uv$, and
its \emph{layers} are the sets $\lambda^{-1}(i)$. For example, the distances to a fixed root vertex of $G$ forms a layering of $G$.

We will also need the notion of \emph{well-colored paths} from~\cite{DumasFPT2024}. Given a graph $G$ whose vertices are covered by $k$ shortest paths, these covering paths are called the \emph{base paths} of $G$ in~\cite{DumasFPT2024}. We orient these base paths from one arbitrary endpoint to the other, and assign  color $i$ to the vertices of the $i$-th base path. Let $\colors(v)$ be the set of colors assigned in this way to vertex $v$.\footnote{In this paper, we apply these definitions to the case where the paths are vertex-disjoint, and thus $\colors(v)$ will be of size~1. However, \Cref{lem:dumas-endpoint-compression} does not require this, so we state it in its general form.} A \emph{coloring} of an arbitrary path $P$ of $G$ is the coloring of the vertices of $P$ obtained by choosing, for every vertex $v$ of $P$, any color in $\colors(v)$. A colored path $P$ is \emph{well-colored} if for each color that occurs in $P$, the vertices of $P$ of this color form a connected monochromatic subpath of $P$. The well-coloring lemma \cite[Lemma 4.1]{DumasFPT2024} ensures the existence of a well-colored, shortest $a$--$b$ path for every $a, b\in V(G)$.

As one of the basic tools of this section, we argue that a slight extension of a construction used to prove \cite[Lemma~4.2]{DumasFPT2024} also yields the following lemma. 

\begin{lemma}\label{lem:dumas-endpoint-compression}
Let $G$ be a connected graph whose vertices are covered by $k$
shortest paths, and let $a\in V(G)$. There exists a family
$\mathcal P$ of $O(3^k)$ well-colored paths starting at $a$ such that every vertex in $V(G)$ belongs to a path in
$\mathcal P$. Moreover, every path $P\in\mathcal P$ ending at a vertex $b$ satisfies $|E(P)|\leq\dist_G(a,b)+2k$.
\end{lemma}

\begin{proof}
We first give an outline of Dumas et al.'s construction~\cite[Section 4]{DumasFPT2024}, and at the end, we explain how we modify their approach. The authors provide a recursive encoding of a set of well-colored paths. The goal of this recursive encoding is to upper-bound the size of the final set by a single-exponential function in $k$ (a simpler counting argument yields an upper bound that is factorial in $k$~\cite[Section 4]{DBLP:conf/isaac/DumasF0T22}). 
The set of paths initially contains, for each vertex $b$ of the layer $\lambda_a^{-1}(i)$ (with $\lambda_a = d_G(a,\cdot)$), a well-colored shortest $a$--$b$ path. At the end of the construction, the set will contain a well-colored $a$--$b$ path, that is not necessarily shortest, but $(1,2k)$-quasi-shortest.

Informally, they define a tree structure, called \emph{$k$-labeled branching tree} \cite[Definition 3.5]{DumasFPT2024}, whose nodes represent sets of paths. The following invariant is preserved throughout the construction: each node $x$ is associated to a well-colored \emph{representative path} $P_x$, defined as the longest path associated to $x$, and such that any other path associated to $x$ is either included in $P_x$, or admits $P_x$ as a prefix. This representative path is defined 
before Step 2.1 of \cite[Lemma 4.2]{DumasFPT2024}.

For any internal node $x$, every path is copied to a child, either unchanged or after undergoing a \emph{canonization} operation that preserves its endpoints (Step 2.3 in~\cite[Lemma 4.2]{DumasFPT2024}). All paths that share the same color $c$ and orientation after their current representative path prefix are grouped together in a child node $x'$ of $x$, and the representative path of $x'$ is obtained by appending to $P_x$ a subpath of the base path whose color is $c$. An additional operation is carried out on the root node, in order to define the representative path of its child nodes (see Steps 1.1 to 1.3 in \cite[Lemma4.2]{DumasFPT2024}). Since all representative paths are well-colored, and each split fixes a new color segment in the representative path, the height of the $k$-labeled branching tree is at most $k$ \cite[Claim 4.4]{DumasFPT2024}. This splitting stops when all paths of a node are included in its representative path. In this case, the considered node is a leaf of the $k$-labeled branching tree.

Furthermore, at depth $d$, the construction maintains the invariant
\begin{equation}
    |E(P)|\leq\dist_G(a,b)+2d\label{eq-invariant}
\end{equation}
for every current $a$--$b$ path $P$ \cite[Claim~4.3]{DumasFPT2024}.\footnote{The claim proves this invariant for vertices at a fixed distance $D$ from $a$. When considering an $a$--$b$ path, we use $D=d_G(a,b)$.} In other words, any $a$-$b$ path associated to a node of depth $d$ is a well-colored $(1,2d)$-quasi-shortest $a$--$b$ path.
The authors finally prove that the number of leaves of a $k$-labeled branching tree is $O(3^k)$ in \cite[Lemma~3.9 and Claim~4.5]{DumasFPT2024}. Note that this whole construction can be started with any set of well-colored shortest paths, all starting at a common vertex $a$.

\medskip

We next prove the lemma statement by applying the same construction to a different starting set of paths. Consider the set $\mathcal Q$ containing, for any vertex $b\in V(G)-a$, a well-colored shortest $a$--$b$ path. (In~\cite{DumasFPT2024}, only paths to vertices of a fixed distance layer were considered.) Recall that the existence of such paths is justified in \cite[Lemma 4.1]{DumasFPT2024}. We create a root node, to which we associate $\mathcal Q$, and we apply the previous construction. By construction, a path towards each vertex $b\in V(G)-a$ is included in one of the representative paths of a leaf of the obtained $k$-labeled branching tree, as the construction preserves path endpoints. Thus, by defining $\mathcal P$ as the set of all representative paths of leaves of the $k$-labeled branching tree, the (endpoints of) paths of $\mathcal P$ cover the vertices of $G$. The bound on the number of leaves yields $|\mathcal P| = O(3^k)$, and the maintained invariant (\ref{eq-invariant}) ensures that $|E(P)|\leq\dist_G(a,b) + 2k$ for all $a$--$b$ paths of $\mathcal P$.
\end{proof}

We next show how to apply \Cref{lem:dumas-endpoint-compression} to a \emph{partition} into $k$ shortest paths, to obtain another path partition with desirable properties.\footnote{Note that \cref{lem:vertex-thin-partition} does not hold when the partition is replaced by a covering: consider for example the graph covered by two paths of the same length, that share every second vertex along them.}

\begin{lemma}\label{lem:vertex-thin-partition}
Let $G$ be a connected graph whose vertices are partitioned into $k$
shortest paths. For every vertex $a\in V(G)$, there is a partition of $V(G)$
into $O(k3^k)$ paths, each containing at most one
vertex at any given distance from $a$.
\end{lemma}

\begin{proof}
By \cref{lem:dumas-endpoint-compression}, there is a set $\mathcal P$ of $O(3^k)$ well-colored paths starting at $a$ that cover $V(G)$, each having length at most
the distance between its endpoints plus $2k$.

For each path $P\in \mathcal P$, consider the edges of $P$ sequentially, from $a$ towards the other endpoint of $P$. Let $xy$ be an edge of $P$ such that $x$ is closer to $a$ than $y$ on $P$. Cut $P$ at edge $xy$ if one of the following two cases occur: (i) if $x$ and $y$ have different colors in the coloring of $P$, and (ii) if $\dist(a,y)\leq \dist(a,x)$, that is, the distance to $a$ fails to increase along $xy$. There are at most $k-1$ color changes in $P$, since it is well-colored.
Moreover, the distance from $a$ changes by at most~1 along each edge.
An edge on which this distance stays constant contributes~1 to the
length excess, and an edge on which it decreases contributes~2, thus there are at
most $2k$ such edges. Hence, $P$ is cut at most $3k-1$ times, producing at most $3k$ subpaths. Each subpath is monochromatic, as we have cut it along all edges whose endpoints have different colors. The distance from $a$ strictly increases along each subpath,
so each piece meets every distance layer at most once. 
Repeating this procedure for every path in $\mathcal P$, let $\mathcal S$ be the family of all obtained subpaths. We have $|\mathcal S|=O(k3^k)$, and
$\mathcal S$ covers $V(G)$ (but may not be a partition). Since each base path is shortest and hence induced, a monochromatic subpath is necessarily a shortest path.

It remains to remove overlaps between the members of $\mathcal S$. Since all paths in $\mathcal S$ are monochromatic and the base paths are pairwise vertex-disjoint, only subpaths of the same base path may overlap. 
Fix such a base path $B$. The members of $\mathcal S$ contained in $B$
are subpaths of $B$ that together cover $V(B)$, but may overlap because they come from different paths supplied by
\cref{lem:dumas-endpoint-compression}. Consider two such overlapping subpaths $P$ and $Q$. If one of them is included in the other, say $V(P)\subseteq V(Q)$, remove $P$ from $\mathcal S$. Otherwise, remove $V(P)\cap V(Q)$ from, say, $P$: this simply shortens $P$. Repeat this until no overlapping paths remain; at the end, we obtain the desired partition of size $O(k3^k)$.
\end{proof}

The next lemma shows how to use a path partition as the one from \Cref{lem:vertex-thin-partition} and a tree decomposition to obtain a path decomposition.

\begin{lemma}\label{lem:vertex-layered-conversion}
Let $G$ be a graph of treewidth at most $w\geq 1$. If $V(G)$ is partitioned into $N\geq1$ paths,
each containing at most one vertex of each layer of a fixed layering $\lambda$, then $\pw(G)=O(w\log(N+1))$.
\end{lemma}

\begin{proof}
Let $\mathcal{P}$ be the given path partition of $V(G)$ and call $f(N)$ the maximum possible pathwidth of a graph partitioned into $N$ paths satisfying the hypotheses of the lemma (treewidth at most $w$ and layering condition). We proceed by induction on $N$.

Suppose first that $N\leq 4w+5$. This means that each layer of the considered layering $\lambda$ is of size at most $4w +5$. Let $\alpha$ be the maximum integer such that $\lambda^{-1}(\alpha) \neq\emptyset$, and define for $i\leq \alpha-1$ the bag $B_i = \lambda^{-1}(i)\cup \lambda^{-1}(i+1)$. The decomposition $B_0, B_1, \dots, B_{\alpha-1}$ is a path decomposition of $G$, since any edge of $G$ has its endpoints in the same layer of $\lambda$, or in two consecutive ones. Thus, $f(N)\leq 8w+9$.

Suppose now that $N> 4w+5$, and let $(T, \mathcal B)$ be a tree decomposition of $G$ of width at most $w$. For each path $P_i$ of $\mathcal P$, we define $T_i$ as the subgraph of $T$ containing all nodes whose bags contain vertices of $P_i$. Let $\mathcal T=\{T_i\mid 1\leq i \leq N\}$. Since $P_i$ is connected, $T_i$ is connected as well. Select, for each $P_i$, a node $x_i$ in $T_i$. Assign to any node $y$ the weight $w(y) = |\{i\mid x_i = y\}|$. The total weight of $T$ is thus $N$. Consider a weighted centroid $c_T$ of $T$, that is, a node of $T$ such that each component of $T-c_T$ has weight at most $\frac N 2$. Call one of these components $C$. If some subtree $T_i$ intersects $C$, since $T_i$ is connected, it is either entirely included in $C$, or $c_T$ is a node of $T_i$. Call $G_C$ the subgraph of $G$ induced by the vertex set $\left(\bigcup_{x\in V(C)}B_x\right)\setminus B_{c_T}$. 

Since $\mathcal P$ is a partition, at most $w+1$ of the subtrees of $\mathcal T$ contain $c_T$ (thus, at most $w+1$ paths of $\mathcal P$ can intersect $B_{c_T}$). Moreover, since the weight of $C$ is at most $\frac N 2$, at most $\frac N 2$ subtrees of $\mathcal T$ are entirely contained in $C$ (they correspond to at most $\frac N 2$ paths of $\mathcal{P}$). Thus, there are at most $\frac N 2 + w + 1$ paths of $\mathcal P$ that intersect $G_C$. We consider the collection $\mathcal P'$ of paths obtained from the intersection of the paths of $\mathcal P$ with $V(G_C)$. Note that a path of $\mathcal P$, if it intersects $B_{c_T}$, may correspond to several paths in $\mathcal P'$. However, each vertex of $B_{c_T}$ may only increase the number of paths in $\mathcal P'$ by~1. Thus, $|\mathcal P'|=N'\leq \frac N 2 + 2w+2$. 
The paths in $\mathcal P'$ also meet each layer of $\lambda$ (when restricted to $G_C$) at most once, and $\lambda$ is still a layering of $G_C$. Because $N> 4w+4$, $N'<N$, and we can consider by induction a path decomposition of $G_C$ of width at most $f(N')\leq f\left(\frac N 2 + 2w+2\right)$.

We do this for each component of $T-c_T$. We order these components arbitrarily and concatenate the path decompositions following this ordering. Then, we add to each bag, the vertices of $B_{c_T}$. It is clear that this forms a valid path decomposition of $G$, of width at most

$$f(N) \leq f\left(\frac N 2+2w+2\right) + w+1.$$

Apply this recursive identity $d=\ceil{\log_2(N+1)}$ times to get

$$f(N) \leq f\left(4w+5 - \frac{2w+2}{2^d}\right) + (w+1)d,$$
which gives $f(N) = O(w\log(N+1))$, by using that $f(N) \leq 8w+9$ when $N\leq 4w+5$.
\end{proof}

We are now ready to assemble the pieces to obtain the main result of this section.

\begin{theorem}\label{thm:vertex-upper}
For any graph $G$, if $V(G)$ is covered by $k$ shortest paths, then $\pw(G)=O(k^3)$. 
\end{theorem}
\begin{proof}
Without loss of generality, consider $G$ to be connected. Let $P_1,\ldots,P_k$ be the shortest paths covering the vertices of $G$. Replace each vertex $v$ by a clique
$F_v=\{(v,i):v\in V(P_i)\}$ and each edge $uv$ by all edges between
$F_u$ and $F_v$. Call the obtained graph $H$. Notice that the paths $P_1,\ldots,P_k$ now partition $V(H)$
and remain shortest paths in $H$. Furthermore, choosing one vertex in each of the created cliques gives an induced copy of $G$ in $H$, so $\pw(G)\leq\pw(H)$.

By \cref{cor:exact-treewidth}, $\tw(H)=O(k^2)$.
\Cref{lem:vertex-thin-partition} applied to $\mathcal P=\{P_1,\ldots,P_k\}$ gives a partition into
$N=O(k3^k)$ paths containing at most one vertex of each distance layer
from a fixed root. The distance from that root is a layering of $H$.
Since $\log(N+1)=O(k)$, \cref{lem:vertex-layered-conversion}
yields $\pw(G)\leq\pw(H)=O(k^3)$.
\end{proof}

\section{Improving the pathwidth upper bound when covering the edges}\label{sec:edge-upper}

We will use the polynomial pathwidth-modulator of
Baste et al.~\cite{BasteDMGOP2026} together with the results in \Cref{sec:treewidth} 
to improve the $O(k^4)$ upper bound for pathwidth from Baste et al.~\cite{BasteDMGOP2026} when the graph is edge-covered by $k$ shortest paths. The results of this section do not apply to quasi-shortest paths. We restate the structural result at the core of the  work in~\cite{BasteDMGOP2026}.

\begin{theorem}[{Baste, De Meyer, Giocanti, Objois, Picavet \cite[Theorem~8]{BasteDMGOP2026}}]\label{thm:baste-pathwidth-components}
    Let $G$ be a graph whose edges are covered by shortest paths $P_1,\ldots,P_k$. For each $i$ with $1\leq i\leq k$, there is a set $X_i$ of at most $720k^3+4k$ vertices of $G$ such that every component of $G-X_i$ meeting $P_i$ has pathwidth at most $6k$.
\end{theorem}

To prove the results in this section, we start from a tree decomposition $(T,\mathcal B)$ of our input graph $G$ given by \Cref{cor:exact-treewidth} or \Cref{cor:minor-closed-treewidth}, and transform it using \Cref{thm:baste-pathwidth-components} and \Cref{lem:binary-tree-minor,lem:pathwidth-leafage-paths} so that the obtained tree $T$ has bounded pathwidth. Then, we use \Cref{lem:semi-tree} below that gives a bound on the pathwidth of $G$ with respect to the pathwidth of $T$. Note that this lemma applies to \emph{semi-tree decompositions}, a generalization of tree decompositions defined in \cite{BodlaenderFomin2002}. In particular, any tree decomposition is a semi-tree decomposition.

\begin{lemma}[Bodlaender and Fomin~{\cite[Lemma~8]{BodlaenderFomin2002}}]
\label{lem:semi-tree}
If a graph $G$ has a semi-tree decomposition $(T, \mathcal B)$ of width $k$, then $\pw(G)\leq (k+1)(\pw(T)+1)-1.$
\end{lemma}

A \emph{complete binary tree of height $h$} is a rooted tree in which every
non-leaf vertex has two children and every leaf is at distance $h$ from the
root.

\begin{lemma}[Marshall and Wood~{\cite[Lemma~6]{MarshallWood2015}}]
\label{lem:binary-tree-minor}
Let $T$ be a forest and let $t\geq 1$.  If $\pw(T)\geq t$, then $T$ contains
a complete binary tree of height $t-1$ as a minor.  $[\dots]$
\end{lemma}

\begin{corollary}\label{cor:pathwidth-bound-leafage}
    Any tree $T$ with $q$ leaves has pathwidth at most $\floor{\log_2 q} + 1$.
\end{corollary}

\begin{proof}
    A tree $T$ with $q$ leaves cannot contain a binary tree of height $\floor{\log_2 q}+1$ as a minor. By contraposition of \Cref{lem:binary-tree-minor}, the pathwidth of $T$ is then at most $\floor{\log_2 q} + 1$.
\end{proof}

\begin{lemma}\label{lem:pathwidth-leafage-paths}
    Let $T$ be a subcubic tree and $T_0$ a subtree of $T$ with $q$ leaves such that $T-T_0$ is a forest of paths. Then, $\pw(T) \leq \floor{\log_2 q}+5$.
\end{lemma}
\begin{proof}
Each path $P$ in $T-T_0$, is connected to $T_0$ by an edge, which can originate from an internal vertex of $P$ or from an endpoint of $P$. Since $T$ is subcubic, at most two such paths are connected to each leaf of $T_0$. Denote by $T'$ the subtree of $T$ obtained from $T_0$ by adding all paths of $T-T_0$ connected by an edge to a leaf of $T_0$. Notice that $T'$ has at most $4q$ leaves. By \Cref{cor:pathwidth-bound-leafage}, $\pw(T') \leq \floor{\log_2(4q)} + 1 = \floor{\log_2 q} + 3$.

Consider a path decomposition $(T^*,\mathcal B^*)$ of $T'$ of width at most $\floor{\log_2 q} + 3$. We show how to transform it into one for $T$ by increasing the bag size by at most~2. To do so, consider the paths of $T-T_0$ connected to internal vertices of $T_0/T'$. Since $T$ is subcubic, at most one path $P_v$ of $T-T_0$ is connected to each internal vertex $v$ of $T_0/T'$. 
Consider the natural path decomposition of width at most~1 of $P_v$. Consider any bag $B$ of $(T^*,\mathcal B^*)$ containing $v$. Make a copy of $B$ and insert, between the two copies of $B$, the sequence of bags of the path decomposition of $P_v$, where we add to each bag, the vertices of $B$. Repeating this for every internal vertex of $T_0$ that has an edge to some path in $T-T_0$, we obtain a path decomposition of $T$ of width at most $\floor{\log_2 q} + 5$, as claimed.
\end{proof}

\begin{lemma}\label{lem:conversion}
    Let $\tw(G)\leq w$ and $X$ be a non-empty vertex set of $G$ of size $r$ such that $\pw(G - X) \leq p$. Then, $\pw(G)\leq (w+p+2)(\floor{\log_2 r}+6)-1$.
\end{lemma}

\begin{proof}
    Take a subcubic tree decomposition $\mathcal T = (T,\mathcal B)$ of $G$ (for instance a classic nice tree decomposition, see~\cite{bodlaender2013fine}) of width $w$, that is, where $T$ is subcubic. For each vertex $x\in X$, fix a bag $B_x$ of $\mathcal B$ containing~$x$.  Define $T_0$ as the inclusion-wise minimal subtree of $T$ containing $B_x$ for all $x\in X$. By minimality, $T_0$ has at most $r$ leaves. Notice that $T-T_0$ is a forest. We call a connected component of $T-T_0$ a \emph{hanging subtree}. $T$ can thus be decomposed into a set of hanging subtrees attached to $T_0$. Let $T'$ be one of those hanging subtrees, call $\alpha$ the node of $T_0$ connected to $T'$, and call $B_\alpha$ its corresponding bag. Denote by $V'$ the set of vertices of $G$ that are present in bags corresponding to nodes of $T'$ except vertices of $X$. Since $\pw(G-X) = p$, we have that $\pw(G[V']) \leq p$. By considering a path decomposition of $G[V']$ and adding vertices of $B_\alpha$ to all bags of this decomposition, we get a path decomposition of $G[V'\cup B_\alpha]$ of width at most $p+w+1$. Replace \(T'\) by the path underlying this path decomposition, and attach its first node to \(\alpha\). Since \(B_\alpha\) is contained in every bag of the path decomposition, tree decomposition properties are preserved across the attachment. Repeat this for each component of $T-T_0$, and call the resulting tree $\bar T$. This tree decomposition has now width at most $w+p+1$ with $\bar T$ being subcubic, and $\bar T - T_0$ is a forest of paths. Thus, we have that $\pw(\bar T) \leq \floor{\log_2 r} + 5$ by \Cref{lem:pathwidth-leafage-paths}. We finally use \Cref{lem:semi-tree} to get that $\pw(G) \leq (w +p+2)(\pw(\bar T) + 1) - 1$, which yields the claim.
\end{proof}

\begin{theorem}\label{thm:edge-upper}
If $E(G)$ is covered by $k\geq1$ shortest paths and $\tw(G)\leq w$, then
 $\pw(G)\leq(w+6k+2)(\floor{4\log_2 k}+16)-1$. 
Consequently, every graph edge-covered by $k$ shortest paths satisfies
$\pw(G)=O(k^2\log k)$. Moreover, for a fixed graph $H$, if $G$ is $H$-minor-free, then $\pw(G)=O_H(k\log k)$.
\end{theorem}

\begin{proof}
We may assume without loss of generality that $G$ is connected and has at least one edge. 
Let
$P_1,\ldots,P_k$ be the covering paths. By \cref{{thm:baste-pathwidth-components}}, for each $i$ with $1\leq i\leq k$, there is a
set $X_i$ of at most $720k^3+4k$ vertices such that every component of
$G-X_i$ meeting $P_i$ has pathwidth at most $6k$.

Let $X=\bigcup_iX_i$.  Then $|X|\leq720k^4+4k^2\leq 724k^4.$
Every component $K$ of $G-X$ contains a vertex of some $P_i$, because the
paths cover all edges and there are no isolated vertices.  Since
$X_i\subseteq X$, the component $K$ lies in a component of $G-X_i$ meeting
$P_i$. Hence, we have $\pw(G-X)\leq6k$.

We apply
\cref{lem:conversion}: the bound on $|X|$ gives $\floor{\log_2 |X|}\leq\floor{4\log_2 k} +10$ and proves the sought inequality. Note that if the paths cover the edges, they also cover the vertices (we disregard isolated vertices as $G$ was assumed to be connected with at least one edge).  We may therefore apply the
above upper bound with $w=\tw(G)$.  By \cref{cor:exact-treewidth},
$\tw(G)\leq441k^2+90k$, and we obtain $\pw(G)=O(k^2\log k)$. For $H$-minor-free graphs, we have $\tw(G)=O_H(k)$ by \cref{cor:minor-closed-treewidth} and thus, $\pw(G)=O_H(k\log k)$.
\end{proof}

\section{A superlinear lower bound for edge-covering by shortest paths}\label{sec:edge-lower}

We now show that existing constructions provide a new superlinear lower bound for the treewidth of graphs coverable with $k$ shortest paths, answering a question of Baste et al.~\cite{BasteDMGOP2026}. Our construction is inspired by the one in~\cite[Theorem~8]{CampbellDGGHHMW2023}.

\begin{theorem}\label{thm:edge-lower}
There is an absolute constant $c>0$ such that, for infinitely many $k$, there is a
cubic graph $G_k$ whose edges are covered by at most $k$ shortest paths and which satisfies $\pw(G_k) \geq \tw(G_k)\geq ck\log k.$
\end{theorem}
\begin{proof}
There exist arbitrarily large cubic bipartite Ramanujan graphs (a Ramanujan graph is a regular graph with specific spectral properties) with
girth $g=\Theta(\log n)$~\cite[Theorem~5.4]{Chiu1992}. Ramanujan graphs are expanders and thus have linear treewidth (because expanders do not have sublinear separators~\cite{HooryLinialWigderson2006}).

Every bipartite cubic graph can be
3-edge-colored (König's Line Coloring Theorem).
Hence, the edges can be partitioned into three
matchings $M_1,M_2,M_3$.  Together, two of them form a collection of disjoint cycles
(because each vertex has two incident edges from them) with $n$ edges in total.
Consider one of these cycles $C$ in $M_1\cup M_2$ or $M_1\cup M_3$ of length $q \geq g=\Omega(\log n)$. $E(C)$ can be covered by
$O(q/g)=O(q/\log n)$ edge-disjoint
paths, each of length less than $g/2$. These paths are automatically shortest paths, because a shortcut would imply the existence of a cycle of length less than the girth, which is impossible. Overall, we have $\Theta(n/\log n)$ shortest paths covering the edges of $M_1\cup M_2$, and the same for $M_1\cup M_3$.

Hence, taking the union of all these shortest paths covering $M_1\cup M_2$ and $M_1\cup M_3$, we have $k=\Theta(n/\log n)$ shortest paths
covering all edges of the graph. Hence, $k\log k=\Theta(n)$, and the treewidth of $G$ is $\Omega(n)$, which is thus $\Omega(k\log k)$.
\end{proof}

\section{Small quasi-shortest path covers can have unbounded pathwidth}\label{sec:quasi-shortest-unbounded}

We now show that when $\alpha>1$, pathwidth cannot be bounded in terms of $k$ when the graph is covered by $(\alpha,\beta)$-quasi-shortest paths, even for $k=2$ and $\beta=0$.

\begin{theorem}\label{thm:twopaths}
For every $\alpha>1$ and every integer $w\geq0$, there is a graph $G$
of order $2^{O_\alpha(w)}$ whose edges are covered by two $(\alpha,0)$-quasi-shortest paths and
whose pathwidth is at least $w$.
\end{theorem}

\begin{proof}
Fix $\alpha>1$. For each $h\geq 0$, we construct a connected graph $G_h$ edge-covered by
two $(\alpha,0)$-quasi-shortest paths $P_h,Q_h$ of equal
length with a common final endpoint.

For $h=0$, take a single vertex and paths that consist of that vertex.  Suppose
$G_{h-1}$ has been constructed, and its two paths have length $d$.
Take two disjoint copies $A,B$ of $G_{h-1}$ and a new vertex $z$.
Choose $l\geq2$
as the smallest integer satisfying the first inequality below, and then
choose $m\geq2$ as the smallest integer satisfying the second:
\begin{equation}\label{equ:lm}
 l+2d\leq\alpha l, \qquad
 m+l+2d\leq\alpha m.
\end{equation}
\begin{figure}[t]
\centering
\includegraphics{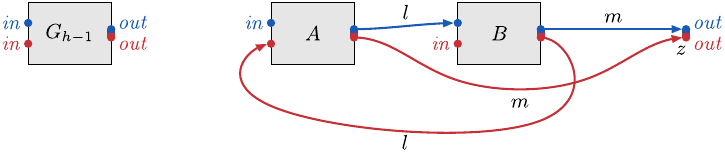}
\caption{Construction of $G_h$ from two copies of $G_{h-1}$.
The blue path is $P_h$ and the red one~$Q_h$.
Labels indicate the lengths of new connecting paths, whose internal
vertices are pairwise distinct.
Blue and red terminals are the endpoints of the
respective covering paths, with inputs on the left and outputs on the
right of each box.  The outputs of the two paths are always
identical.}
\label{fig:twopaths}
\end{figure}
Join the corresponding path endpoints according to Figure~\ref{fig:twopaths}
by new paths of length $l$ or $m$ to obtain
\[
 P_h=P_A\xrightarrow{\,l\,}P_B\xrightarrow{\,m\,}z,\text{ and }
 Q_h=Q_B\xrightarrow{\,l\,}Q_A\xrightarrow{\,m\,}z.
\]
Let $G_h=P_h\cup Q_h$.  Both $P_h$ and $Q_h$ have length $2d+l+m$
and final endpoint~$z$.

We first prove quasi-shortness of $P_h$ and $Q_h$.  For any two distinct vertices
$x,y$ lying on both $P_h$ and $Q_h$, we claim that their distance along these paths
differs by a factor of at most $\alpha$.  If $x,y$ belong to the same
copy $A$ or~$B$, this follows from induction.  If they belong to
$A$ and $B$, both distances lie in $[l,l+2d]$.  If one is $z$
and the other belongs to $A$ or $B$, both distances lie in $[m,m+l+2d]$.
(\ref{equ:lm})~proves the quasi-shortness condition of the subpaths of $P_h$ or $Q_h$ with
endpoints $x$ and~$y$.

Now let $x,y\in V(P_h)$, and take a shortest $x$--$y$ path $R$ in $G_h$.
Every excursion of $R$ outside $P_h$ is a subpath of $Q_h$ with endpoints
common to both covering paths.  Replace each excursion by the
corresponding subpath of $P_h$ that is at most by a factor $\alpha$ longer
because of the preceding paragraph.
The resulting walk is contained in $P_h$ and has length
at most $\alpha|R|$ and therefore $\dist_{P_h}(x,y)\leq\alpha\cdot\dist_{G_h}(x,y).$
Thus $P_h$ is $(\alpha,0)$-quasi-shortest
and the same argument applies to $Q_h$.

Finally, deleting the two length-$l$ connecting paths,
but retaining their endpoints, at every recursive level,
leaves a subdivision of the complete binary tree of height $h$, rooted at
the common final endpoint.  The case $h=0$ is a single vertex,
and at each step the two length-$m$ connecting paths join the new root $z$
to the roots of the subdivisions in $A$ and~$B$.
According to \cite[Theorem~67]{Bodlaender1998Arboretum}, the complete binary tree of
height $h$ has pathwidth~$\lceil h/2\rceil$. Thus, $G_h$ contains the complete binary tree of height~$h$ as a minor, and since pathwidth is
minor-monotone, taking $h=2w$ gives us the required graph~$G_h$ with
pathwidth at least~$w$. By our choices of $l$ and $m$, both are in $O_\alpha(d)$, and therefore $d$ grows geometrically with every recursive construction, implying that $G_h$ has order $2^{O_\alpha(w)}$.
\end{proof}

Exponential size is optimal up to the constant in the exponent:
by \cref{thm:treewidth},
every graph covered by two quasi-shortest paths has constant treewidth.
Since every $n$-vertex graph of treewidth $t>0$ has pathwidth
$O(t\log n)$~\cite{Bodlaender1998Arboretum}, pathwidth~$w$ implies $n=2^{\Omega_\alpha(w)}$.

\section{XP recognition for quasi-shortest path/walk covers}\label{sec:blocked}

\emph{Counting Monadic Second-Order} (CMSO) logic refers to \emph{Monadic Second-Order} (MSO) logic augmented with fixed modular-cardinality predicates (that is, checking whether a given set has size $a$ modulo $b$, for fixed values $a,b$). When considering graphs as inputs,
$\mathrm{CMSO}_1$ quantifies over vertex sets, while $\mathrm{CMSO}_2$ also
quantifies over edge sets.  It can express paths, incidence, and connectivity, but
cannot compare the cardinalities of two quantified sets.  We recall the precise $\mathrm{CMSO}_2$ fragment of Dreier, Ganian, and Hamm~\cite[Section~3.3 and
Definition~2]{DreierGH2025} that we will use. To do so, we fix nonnegative integer vertex and edge weight
symbols $w_i$ (given as inputs along with the graph).

A \emph{linear weight term} has the form $a+\sum_{i,j}a_{ij}w_i(X_j)$, where
 $w_i(X)=\sum_{x\in X}w_i(x)$, and the terms $a$ and $a_{ij}$ are constants called \emph{coefficients}. A coefficient is
\emph{nonnegative} if $a$ and all $a_{ij}$ are nonnegative real numbers, and
\emph{natural} if they are all natural numbers, including zero.  The logic
$\mathrm{CMSO}[\lessgtr]$
extends classical CMSO~\cite{Courcelle1990} by allowing atomic comparisons between
weight terms.

A \emph{block formula} has the form
$\forall\bar Y\,\phi(\bar X,\bar Y)$, where $\phi$ is a
$\mathrm{CMSO}[\lessgtr]$ formula containing only nonnegative weight terms.
All but at most one exceptional weight term depend only on the outer variables $\bar X$.
This exceptional term, if present, may also depend on $\bar Y$, but it occurs
only on one side of a single comparison.  \emph{Blocked
CMSO}, denoted $\forall\mathrm{CMSO}$, is the closure of block formulas under
conjunction, disjunction, and existential first- and second-order
quantification.  
For the sake of readability, we will slightly abuse notation when quantifiers and weight functions are applied to single variables. In particular, we will write $\forall x$ instead of $\forall x'\in \{x\}$, and $w(x)$ instead of $w(\{x\})$.

A $\forall\mathrm{CMSO}$ \emph{query} pairs a constraint using natural weight
comparisons with an integral linear optimization target. An \emph{assignment} for a query is a set of values for the free variables involved in it, and it is \emph{satisfying} if the query constraints hold when using the assigned values. Every query used
here is a feasibility query with the constant-zero target; it returns either a
satisfying assignment or a report that none exists. 

\begin{lemma}\label{lem:blocked-path}
For every fixed rationals $\alpha\geq1$ and $\beta\geq0$, there is a $\forall\mathrm{CMSO}_2$ formula
$\varphi_{\alpha,\beta}(P)$ with one free edge-set variable $P$ such
that $G\models\varphi_{\alpha,\beta}(P)$ if and only if $P$ is the
edge set of a positive-length $(\alpha,\beta)$-quasi-shortest path.
\end{lemma}

\begin{proof}
Let $\alpha=\frac{p}{q}$ and $\beta=\frac{r}{s}$ with nonnegative integers $p,q,r,s$. Let
$\operatorname{Path}(P)$ be the standard $\mathrm{MSO}_2$
predicate that $P$
is the edge set of a nonempty positive-length simple path. It asks that $P$ forms a connected subgraph, no vertex of $G$ is incident to three edges of $P$, and $P$ is acyclic (every subset of $P$ has a leaf), see \cite[Chapter 7]{CyganFKLMPPS2015}. Moreover, let
$\operatorname{RootedPath}(R,x,y)$ say that $R$ is the edge set of an $x$--$y$
path (same as $\operatorname{Path}(R)$ but $x,y$ are distinct and have degree~1 when considering only edges of $R$), $\operatorname{AllEdges}(F)$ say that $F=E(G)$, and
$\operatorname{Complement}(S,F,R)$ say that $S=F\setminus R$.  These are all
ordinary $\mathrm{MSO}_2$ predicates.

Let $\operatorname{SubPath}(Q,P,x,y)$ be the predicate saying that $Q\subseteq P$ is a subpath of $P$ with endpoints $x,y$. It is expressed as $\operatorname{RootedPath}(Q,x,y)$ together with the fact that $Q\subseteq P$. Hence, $\operatorname{SubPath}(Q,P,x,y)$ is $\mathrm{MSO}_2$.

For a subpath $Q$ with endpoints $x$ and $y$, quasi-shortness requires
$|Q|\leq \alpha|R|+\beta=\frac{p}{q}|R|+\frac{r}{s}$ for any shortest $x$--$y$ path $R$ (and thus, for any $x$--$y$ path $R$), that is, $sq|Q|\leq sp|R|+rq$. Blocked CMSO does not allow
this comparison because both sides depend on universally quantified sets $Q$ and~$R$. 
However, by using $S=E(G)\setminus R=F\setminus R$, we get $|R|=|F|-|S|$, and so the above condition can be reformulated as
\begin{equation}\label{eq:blocked-comparison}
sq|Q|+sp|S|\leq sp|F|+rq,
\end{equation}
which is allowed in Blocked CMSO because only $Q$ and $S$ are universally quantified (since $F=E(G)$ is fixed), and these universally quantified terms appear only on one side of the comparison. This is a crucial point.
Formally, let $\varphi_{\alpha,\beta}(P)$ be the following query:
\[
\begin{split}
\operatorname{Path}(P) \land \exists F\bigl(\operatorname{AllEdges}(F)\land
\forall Q\,\forall R\,\forall S\,\forall x\,\forall y\,[
 &\neg\operatorname{SubPath}(Q,P,x,y)\\
 {}\lor{}&\neg\operatorname{RootedPath}(R,x,y)\\
 {}\lor{}&\neg\operatorname{Complement}(S,F,R)\\
 {}\lor{}&sq|Q|+sp|S|\leq sp|F|+rq]\bigr).
\end{split}
\]

Requiring \eqref{eq:blocked-comparison}
for every $x$--$y$ path $R$ and subpath $Q$ is equivalent to
$|Q|\leq \alpha\cdot\dist_G(x,y)+\beta$. Universally quantifying $Q,x,y$
therefore gives exactly \cref{def:quasi-shortest}.

In $\varphi_{\alpha, \beta}(P)$, $P$ and $F$ are fixed outside the universal
block.  Hence $sp|F|+rq$ is an outer term, while $sq|Q|+sp|S|$ is the unique term
depending on block variables and occurs on only one side of one comparison.
Both are nonnegative natural weight terms, and all predicates are pure MSO$_2$.
Thus the formula belongs to $\forall\mathrm{CMSO}_2$.
\end{proof}

\begin{lemma}\label{lem:blocked-cover}
For fixed rationals $\alpha\geq 1$ and $\beta\geq 0$, and every $k\geq1$, there is a
$\forall\mathrm{CMSO}_2$ formula $\widehat\psi_{k,\alpha,\beta} (X_1,\ldots,X_k,P_1,\ldots,P_k)$ 
of size $O(k)$ whose satisfying assignments encode precisely the vertex
covers by at most $k$ $(\alpha,\beta)$-quasi-shortest paths.
\end{lemma}

\begin{proof}
For every label $i$, the vertex set $X_i$ and edge set $P_i$ encode one
path.  An ordinary MSO$_2$ predicate distinguishes two cases.  If $P_i$ is
empty, then $X_i$ is empty or a singleton.  If $P_i$ is nonempty, then $X_i$
is exactly the set of vertices incident with $P_i$, and
$\varphi_{\alpha,\beta}(P_i)$ holds.  Finally, require
$X_1\cup\cdots\cup X_k=V(G)$.  Empty labels allow fewer than $k$ paths.
The obtained formula
is closed under the required conjunctions, disjunctions, and
existential quantifiers, and each of the $k$ copies has constant size for
fixed $\alpha,\beta$.
\end{proof}

We use the following XP algorithmic meta-theorem to get the main result of this section.

\begin{theorem}[{Dreier, Ganian, Hamm~\cite[Theorem~5]{DreierGH2025}}]\label{thm:metatheorem}
    There exists a computable function $h$ such that for any $\forall\mathrm{CMSO}_2$ query $\psi$ with maximum coefficient or input weight $M$, and any graph $G$ of treewidth $w$ and order $n$, one can check in time $(M+n)^{h(|\psi|,w)}$ whether $\psi$ is satisfied by $G$. If yes, one can also obtain the corresponding values of the free variables in $\psi$.
\end{theorem}

We apply \cref{thm:metatheorem} to obtain the following.

\begin{theorem}\label{thm:xp}
For every fixed rationals $\alpha\geq1$ and $\beta\geq0$, there is a computable function
$f_{\alpha,\beta}$ such that the following
holds.  Given a simple graph $G$ with $n>0$ vertices and an
integer $k\geq1$, one can solve \textsc{$(\alpha,\beta)$-Quasi-Shortest Path Cover}
in XP time $n^{f_{\alpha,\beta}(k)}$.
\end{theorem}

\begin{proof}
Let $a=2\alpha+\beta+1$ and 
 $w=49a^2k^2+30ak$. 
Apply the single-exponential 2-approximation for treewidth due to
Korhonen~\cite{Korhonen2022}.  It either correctly reports
$\tw(G)>w$, in which case we answer no due to \cref{thm:treewidth}, or returns a tree decomposition
of width at most $w'=2w+1$.  The standard nice tree decomposition conversion transforms it into a nice tree decomposition of at most $4n$ nodes and width at most $w'$ in time $O(wn)$~\cite[Lemma 6]{bodlaender2013fine}.

By \cref{lem:blocked-cover}, the formula
$\widehat\psi_{k,\alpha,\beta}$ has size $O(k)$ and its assignments are
exactly the required quasi-shortest path covers.  Evaluate the corresponding feasibility query
exactly on the width-$w'$ decomposition using \cref{thm:metatheorem}.  If it is infeasible, answer no.  Otherwise, return the paths
encoded by the free sets $(X_i,P_i)$.

With $\alpha=\frac p q$ and $\beta= \frac r s$, the weight bound from \cref{thm:metatheorem} is the constant
$M=\max\{sq,sp,rq\}$.  Exact evaluation therefore takes time
$\bigl(M+(w+1)n\bigr)^{h(O(k),2w+1)}$.

The treewidth approximation takes $2^{O(w)}n$ time, the nice tree decomposition conversion takes
$O(wn)$ time, and writing the $k$ paths takes $O(kn)$ time.  Since
$w=O_{\alpha,\beta}(k^2)$ by \cref{thm:treewidth} and $M$ is constant for fixed
$\alpha,\beta$, these costs have the claimed form, by aggregating all
parameter-only factors into the exponent
$f_{\alpha,\beta}$.
\end{proof}

We next prove that the XP algorithm from \Cref{thm:xp} can be extended to walks, by a reduction to a slightly constrained version of \textsc{$(\alpha,\beta)$-Quasi-Shortest Path Cover}.

\begin{theorem}\label{thm:walk-xp}
For every fixed rationals $\alpha\geq1$ and 
$\beta\geq0$, there is a computable function
$f_{\alpha,\beta}$ such that the following
holds.  Given a simple graph $G$ with $n>0$ vertices and an
integer $k\geq1$, one can solve \textsc{$(\alpha,\beta)$-Quasi-Shortest Walk Cover} in XP time
$n^{f_{\alpha,\beta}(k)}$.
\end{theorem}

\begin{proof}
Notice that, if $\beta<1$, any $(\alpha,\beta)$-quasi-shortest walk is actually a simple path, because repeating a vertex (which is at distance~0 from itself) would violate the $(\alpha,\beta)$-quasi-shortness condition. Hence, in that case, the statement follows from \Cref{thm:xp}.

Assume next $\beta\geq 1$ and let $b=\lfloor \beta\rfloor+1$. Construct a graph $H$ from $G$ by replacing every vertex $u$ by a
clique $C_u$ of size $b$ and 
adding all possible edges connecting a vertex of $C_u$ to a vertex of $C_v$,
whenever $uv\in E(G)$. Let $\mathcal C=\{C_u\mid u\in V(G)\}$. Color all the edges inside replacement cliques.
We will use a slightly modified version of the formulas in the proof of Theorem~\ref{thm:xp}. To this end, we show that $G$ has a vertex-covering by $k$ $(\alpha,\beta)$-quasi-shortest walks if and only if $H$ has a set $\mathcal P=\{P_1,\ldots,P_k\}$ of at most $k$
$(\alpha,\beta)$-quasi-shortest (simple) paths satisfying the following three conditions:

\begin{enumerate}
    \item none of the paths in $\mathcal P$ uses a colored edge;
    \item no subpath of any path in $\mathcal P$ with at least $b$ edges has both endpoints in the same clique of $\mathcal C$;
    \item every clique of $\mathcal C$ is intersected by some path in $\mathcal P$.
\end{enumerate}

We now show that $\mathcal P$ indeed encodes exactly the desired walk covers of $G$. Consider a set of $(\alpha,\beta)$-quasi-shortest walks of $G$ covering $V(G)$. We may assume without loss of generality, that the same vertex does not occur consecutively in the walk, as otherwise, we may remove one of these occurrences. Notice that in an $(\alpha,\beta)$-quasi-shortest walk $W$ of $G$, the first and last occurrences of any vertex are fewer than $b$ steps apart, so each vertex occurs at most $b$ times. Thus, replacing $W$ by a path $P$ of $H$ by assigning distinct vertices of the same clique $C_u$ of $\mathcal C$ to the occurrences of $u$ in $W$ turns $W$ into a simple path in~$H$. For the same reason, no subpath of such a path with at least $b$ edges has both endpoints in the same clique of $\mathcal C$, so (2) is satisfied. Since $W$ did not contain two consecutive occurrences of the same vertex, (1) is satisfied. Every clique of $\mathcal C$ is intersected since the walks cover all vertices of $G$. Finally, for vertices $x\in C_u$, $y\in C_v$ with $u\ne v$, we have $\dist_H(x,y)=\dist_G(u,v)$, so, a subpath between $x$ and $y$ is $(\alpha,\beta)$-quasi-shortest since the walks are as well. If $x,y$ are in the same clique $C_v$, $\dist_H(x,y)=1$, and as we proved (2), it has at most $b-1\leq\beta\leq\alpha\cdot \dist_H(x,y)+\beta$ edges, so it is also $(\alpha,\beta)$-quasi-shortest.

Conversely, let $\mathcal P$ be a set of $(\alpha,\beta)$-quasi-shortest (simple) paths of $H$ satisfying all three conditions. In each path, we replace each vertex of say $C_v$ by $v$, to obtain a collection of walks of $G$ of the same size as $\mathcal{P}$. Condition (2) guarantees that a subpath with at least $b$ edges in any path of $\mathcal P$ is projected into a subwalk in $G$ with distinct endpoints, hence subwalks with the same starting and ending vertex are $(\alpha,\beta)$-quasi-shortest. For vertices $x\in C_u$, $y\in C_v$ with $u\ne v$, we have $\dist_H(x,y)=\dist_G(u,v)$, so, subwalks with distinct endpoints are $(\alpha,\beta)$-quasi-shortest since the paths of $\mathcal P$ do as well.
Since every clique in $\mathcal C$ is intersected by a path of $\mathcal P$, every vertex of $G$ is covered by one of the walks.

For fixed $\beta$, the three additional conditions are
ordinary MSO$_2$ predicates.  Combine them with the path
formulas of \cref{lem:blocked-path}, treating empty and
singleton paths as in \cref{lem:blocked-cover}.  This gives
a $\forall\mathrm{CMSO}_2$ query of size $O(k)$ on the colored
graph~$H$. As in the proof of \Cref{thm:xp}, by \cref{thm:walk-treewidth} 
and running
Korhonen's treewidth 2-approximation~\cite{Korhonen2022}, we may reject if the treewidth is
too large for a yes-instance, or obtain
a tree decomposition of $G$ of width $w=O_{\alpha,\beta}(k^2)$.
Replacing each bag $B$ by $\bigcup_{u\in B}C_u$ gives a decomposition
of $H$ of width $O(b w)$.  Moreover, $|V(H)|=b n$.
Adapting the proof of \Cref{thm:xp}, Exact Blocked-CMSO evaluation
therefore has the stated XP running time. 
\end{proof}

\section{Conclusion}\label{sec:discussion}

For graphs $G$ covered by $k$ exact shortest paths, our results establish that $\tw(G)= O(k^2)$ and $\pw(G)
 = O(k^3)$. 
For covering the edges, the treewidth upper bound remains, and we obtained $\pw(G) = O(k^2\log k)$.
For all these cases, we provide an $\Omega(k\log k)$ lower bound for infinitely many graphs, showing that a linear upper bound is impossible (negatively answering a question from Baste et al.~\cite{BasteDMGOP2026}).

The polynomial pathwidth upper bound for the vertex-covering case answers the most intriguing formerly open question in the area, since previously, only an exponential upper bound was known. The other bounds improve those from the literature.

 The gaps between the lower and upper bounds come from the quadratic treewidth upper bound, and the additional cost from the pathwidth-modulator and the well-colored path partition, respectively. Note that improving the $\Omega(k\log k)$ lower bound in treewidth would also improve the treewidth lower bound with respect to linear growth from~\cite[Theorem 8]{CampbellDGGHHMW2023}.

The main open question is to obtain tight bounds for the above problems.

\begin{openproblem}
What are the maximum treewidth and pathwidth of a graph covered by $k$ shortest paths? Is $O(k^2)$ optimal, or even $O(k\log k)$? Do the optimal bounds for vertices and edges differ? What about $H$-minor-free graphs?
\end{openproblem}

We introduced the problems \textsc{$(\alpha,\beta)$-Quasi-Shortest Path/Walk Cover}, which can arguably better model real-world applications. Vertex-covering by quasi-shortest paths/walks also implies quadratic treewidth, enabling XP recognition and construction algorithms. The use of the Blocked-CMSO-based meta-theorem yields XP bounds whose exponents depend on the formula and treewidth, but explicit exponents and FPT algorithms remain open. 

\begin{openproblem}
For fixed rationals $\alpha\geq1$ and $\beta\geq0$, are
\textsc{$(\alpha,\beta)$-Quasi-Shortest Path Cover} and
\textsc{$(\alpha,\beta)$-Quasi-Shortest Walk Cover} FPT
when parameterized by the treewidth, or by the number $k$ of solution paths?
\end{openproblem}

Note that this is already open for \textsc{Shortest Path Cover} ($\alpha=1$ and $\beta=0$)~\cite{DumasFPT2024,FernauFMMPR2025}. 
The related \textsc{Shortest Path Partition} problem is W[1]-hard for pathwidth~\cite{ChakrabortyDFMT2026} (but the solution size parameterization remains open). The edge-disjoint version of \textsc{Shortest Path Cover} for oriented graphs is W[1]-hard for solution size~\cite{FernauFMMPR2025}.

\paragraph*{Disclosure of generative-AI assistance.}
We used generative AI (GPT 5 and 6) for initial language drafting, literature search and review of our candidate proofs. In particular, it assisted with examining the applicability of cited results, suggesting alternative arguments, and checking inequalities and boundary cases. These were treated as suggestions and carefully verified by us. The core ideas, proofs and final manuscript are due to the authors: we entirely reviewed, rewrote and verified proofs, citations, final text and retain full responsibility for the content.

\bibliography{refs}
\end{document}